\documentclass{llncs}

\usepackage{amssymb,amsmath,amscd,latexsym}
\usepackage{algorithm}
\usepackage{luca}
\usepackage{url}
\usepackage{pstricks}
\usepackage{pst-node}
\usepackage{listings}
\usepackage{graphicx}
\usepackage{gastex}
\usepackage{subfigure}
\def\qed{\rule{0.4em}{1.4ex}}

\newcommand{\pat}{\pi}
\newcommand{\Path}{\Pi}

\newcommand{\Prb}{\mathbb{P}}
\newcommand{\Inf}{\mathrm{Inf}}
\newcommand{\Reach}{{\text{\textrm{Reach}}}}

\newcommand{\Pat}{\Pi}
\newcommand{\Buchi}{\textrm{B\"uchi}}
\newcommand{\Safe}{\textrm{Safe}}
\newcommand{\coBuchi}{\textrm{coB\"uchi}}
\newcommand{\LimitAvg}{\textrm{LimitAvg}}

\newcommand{\Exp}{\mathbb{E}}

\newcommand{\nats}{\mathbb{N}}

\newcommand{\set}[1]{\{\: #1 \:\}}

\newcommand{\vare}{\varepsilon}

\newcommand{\ov}{\overline}

\newcommand{\wh}{\widehat}

\newcommand{\slopefrac}[2]{\leavevmode\kern.1em
  \raise .5ex\hbox{\the\scriptfont0 #1}\kern-.1em
  /\kern-.15em\lower .25ex\hbox{\the\scriptfont0 #2}}
\newcommand{\half}{\slopefrac{1}{2}}

\begin{document}
\title{The Decidability Frontier for Probabilistic Automata on Infinite Words\protect\footnote[1]{This research was supported 
by the European Union project COMBEST, the European Network of Excellence
ArtistDesign, the Austrian Science Fund (FWF) NFN Grant No S11407-N23 (RiSE), and a Microsoft Faculty Fellowship}$^,$\protect\footnote[2]{This paper is an extension of the article \emph{Probabilistic Automata on Infinite Words: Decidability and Undecidability Results}, which appeared in the Proceedings of ATVA'2010. pp.1~16.}}
\author{Krishnendu Chatterjee \and Thomas A. Henzinger \and Mathieu Tracol}
\institute{IST Austria (Institute of Science and Technology Austria)}
\maketitle
\newif
  \iflong
  \longfalse
\newif
  \ifshort
  \shorttrue

\thispagestyle{empty}
\vspace{-2em}
\begin{abstract}
We consider probabilistic automata on infinite words with 
acceptance defined by safety, reachability, B\"uchi, coB\"uchi, and 
limit-average conditions.
We consider quantitative and qualitative decision problems. 
We present extensions and adaptations of proofs for probabilistic 
finite automata and present a complete characterization of 
the decidability and undecidability frontier of the 
quantitative and qualitative decision problems for probabilistic automata 
on infinite words.
\end{abstract}
\newcommand{\alphab}{\Sigma}
\newcommand{\Cone}{\mathrm{Cone}}
\newcommand{\ii}{\iota}
\newcommand{\first}{\mathsf{first}}

\section{Introduction}

\noindent{\bf Probabilistic automata and decision problems.}
Probabilistic automata for finite words were introduced in the 
seminal work of Rabin~\cite{RabinProb63}, and have been extensively studied 
(see the book by~\cite{PazBook} on probabilistic automata and the survey 
of~\cite{Bukharaev}).
Probabilistic automata on infinite words have been studied recently 
in the context of verification~\cite{BG05,BBG08}.
We consider probabilistic automata on infinite words with 
acceptance defined by safety, reachability, B\"uchi, 
coB\"uchi, and limit-average conditions.
We consider the \emph{quantitative} and \emph{qualitative} decision 
problems~\cite{PazBook,GO-Tech}.
The quantitative decision problems ask, given a rational $0 \leq 
\lambda \leq 1$, whether
(a) \emph{(equality)} 
there is a word with acceptance probability 
exactly $\lambda$;
(b) \emph{(existence)} there is a word with acceptance probability 
greater than $\lambda$; and 
(c) \emph{(value)} for all $\vare>0$, 
there is a word with acceptance probability greater than $\lambda-\vare$.
The qualitative decision problems are the special cases of the 
quantitative problems with $\lambda \in \set{0,1}$.
The qualitative and quantitative decision problems for probabilistic 
automata are the generalization of the emptiness and universality problem 
for deterministic automata.

\smallskip\noindent{\bf Known results for probabilistic automata on infinite
words.}
The decision problems for probabilistic automata on finite words have been 
extensively studied~\cite{PazBook,Bukharaev}.
For probabilistic automata on infinite words it follows from the 
results of~\cite{BBG08} that for the coB\"uchi acceptance condition, 
the qualitative equality problem is undecidable and the qualitative existence
problem is decidable (in EXPTIME), whereas for the B\"uchi acceptance condition, the 
qualitative equality problem is decidable (in EXPTIME) and the qualitative existence
problem is undecidable. 
The qualitative equality and existence problems for safety conditions were shown to 
be PSPACE-complete in~\cite{chadha2009expressiveness}, and the qualitative equality 
problem for B\"uchi conditions and qualitative existence problem for coB\"uchi 
conditions were shown to be PSPACE-complete in~\cite{chadha2009power}.  
The arithmetic hierarchy characterization of the the undecidable problems was 
established in~ \cite{chadha2009expressiveness,chadha2009power}.

\smallskip\noindent{\bf Our results: quantitative decision problems.} 
In~\cite{GO-Tech} a simple and elegant proof for the undecidability of the 
quantitative decision problems for probabilistic \emph{finite} automata 
was given.
In Section~3 we show that the proof of~\cite{GO-Tech} can be extended to show 
that the quantitative decision problems are undecidable for  safety, 
reachability, B\"uchi, coB\"uchi, and limit-average conditions.
In particular we show the undecidability of the quantitative equality 
problem for the special classes of absorbing probabilistic automata with safety and 
reachability acceptance conditions.
The other undecidability results for quantitative decision problems are 
obtained by simple extensions of the result for quantitative equality.

\smallskip\noindent{\bf Our results: qualitative decision problems.}
In Section~4 we show that all  qualitative decision problems are 
decidable for probabilistic safety automata. 
We present a simple adaptation of a proof of~\cite{GO-Tech} to give
the precise characterization of the decidability and undecidability 
frontier for all qualitative decision problems for probabilistic 
reachability, B\"uchi, and coB\"uchi automata.
Concerning probabilistic limit-average automata, we show that 
the qualitative value problem is undecidable. We use the results of \cite{tracol2009recurrence} to prove that the qualitative equality is also undecidable, and that the qualitative existence is decidable.

\section{Definitions}
In this section we present definitions for probabilistic automata, notions
of acceptance for infinite words, and the decision problems.

\subsection{Probabilistic automata}

\smallskip\noindent{\bf Probabilistic automata.} 
A \emph{probabilistic automaton} $\cala$ is a tuple 
$(Q,q_\ii,\alphab,\calm)$ that consists of 
the following components:
\begin{enumerate}
\item a finite set $Q$ of states and an initial state $q_\ii$; 
\item a finite alphabet $\alphab$; 
\item a set $\calm=\set{M_\sigma \mid \sigma \in \alphab}$ of transition 
probability matrices $M_\sigma$; i.e., for $\sigma \in \alphab$ we have 
$M_\sigma$ is a transition probability matrix.
In other words, for all $\sigma \in \alphab$ the following conditions hold: 
(a) for all $q,q' \in Q$ we have $M_\sigma(q,q') \geq 0$; and 
(b) for all $q \in Q$ we have $\sum_{q' \in Q} M_\sigma(q,q') =1$.
\end{enumerate}

\smallskip\noindent{\bf Infinite paths and words.}
Given a probabilistic automaton $\cala$, an infinite path 
$\pat=(q_0,q_1,q_2,\ldots)$ is an infinite sequence of states such 
that for all $i \geq 0$, there exists $\sigma \in \alphab$ 
such that $M_\sigma(q_i,q_{i+1}) >0$.
We denote by $\pat$ a path in $\cala$, and by $\Pat$ the 
set of all paths.
For a path $\pat$, we denote by $\Inf(\pat)$ the set of 
states that appear infinitely often in $\pat$.
An infinite (resp. finite) word is an infinite (resp. finite) sequence of 
letters from $\alphab$.
For a finite word $w$ we denote by $|w|$ the length of $w$.
Given a finite or an infinite word $w$, we denote $w_i$ 
as the $i$-th letter of the word (for a finite word $w$ we assume 
$i \leq |w|$).

\smallskip\noindent{\bf Cones and probability measure.} 
Given a probabilistic automaton $\cala$ and a finite sequence 
$\ov{q}=(q_0,q_1,\ldots,q_n)$ of states, the 
set $\Cone(\ov{q})$ consists of the set of paths $\pat$ with 
prefix $\ov{q}$.
Given a word $w \in \alphab^\omega$, we first define a measure 
$\mu^w$ on cones as follows: 
\begin{enumerate}
\item $\mu^w(\Cone(q_\ii))=1$ and for $q_i \neq q_{\ii}$ we have
$\mu^w(\Cone(q_i))=0$;

\item for a sequence $\ov{q}=(q_0,q_1,\ldots,q_{n-1},q_n)$ 
we have 
\[
\mu^w(\Cone(\ov{q}))= \mu^w(\Cone(q_0,q_1,\ldots,q_{n-1})) \cdot 
M_{w_n}(q_{n-1},q_n).
\]
\end{enumerate}
The probability measure $\Prb^w_\cala$ is the unique extension 
of the measure $\mu^w$ to the set of all measurable paths in 
$\cala$.

\subsection{Acceptance conditions}

\smallskip\noindent{\bf Probabilistic automata on finite words.}
A probabilistic \emph{finite} automaton consists of a probabilistic 
automaton and a set $F$ of final states. 
The automaton runs over finite words $w \in \alphab^*$ and for a 
finite word $w=\sigma_0 \sigma_1 \ldots \sigma_n \in\alphab^*$ 
it defines a probability distribution over $Q$ as follows: 
let $\delta_0(q_\ii)=1$ and for $i \geq 0$ we have 
$\delta_{i+1}=\delta_i \cdot M_{\sigma_i}$.
The acceptance probability for the word $w$, denoted as 
$\Prb_\cala(w)$, is $\sum_{q \in F} \delta_{|w|+1}(q)$.

\smallskip\noindent{\bf  Acceptance for infinite words.} 
Let $\cala$ be a probabilistic automaton and let $F \subseteq Q$ be a 
set of accepting (or target) states. 
Then we consider the following functions to assign values to paths.
\begin{enumerate}

\item \emph{Safety condition.} 
The safety condition $\Safe(F)$ defines the set of paths in $\cala$ 
that only visits states in $F$; i.e., 
$\Safe(F) =\set{(q_0,q_1,\ldots) \mid \forall i \geq 0. \ q_i \in F}$.

\item \emph{Reachability condition.}
The reachability condition $\Reach(F)$ defines the set of paths in 
$\cala$ that visits states in $F$ at least once; i.e., 
$\Reach(F) =\set{(q_0,q_1,\ldots) \mid \exists i \geq 0. \ q_i \in F}$.

\item \emph{B\"uchi condition.} 
The B\"uchi condition $\Buchi(F)$ defines the set of paths in 
$\cala$ that  visits states in $F$ infinitely often; i.e., 
$\Buchi(F) =\set{\pat \mid \Inf(\pat) \cap F \neq \emptyset}$.

\item \emph{coB\"uchi condition.}
The coB\"uchi condition $\coBuchi(F)$ defines the set of paths in 
$\cala$ that  visits states outside $F$ finitely often; i.e., 
$\coBuchi(F) =\set{\pat \mid \Inf(\pat) \subseteq F}$.

\item \emph{Limit-average condition.} The limit-average condition 
is a function $\LimitAvg: \Path \to \reals$ that assigns to 
a path the long-run average frequency of the accepting states.
Formally, for a state $q \in Q$, let $r(q)=1$ if $q \in F$ and $0$ otherwise,
then for a path $\pat=(q_0,q_1,\ldots)$ we have
\[
\LimitAvg(\pat)=\lim\inf_{k \to \infty} \frac{1}{k}\cdot 
\sum_{i=0}^{k-1} r(q_i).
\]
\end{enumerate}
In sequel, we will consider $\Reach,\Safe,\Buchi,\coBuchi$ and 
$\LimitAvg$ as functions from $\Path$ to $\reals$. 
Other than $\LimitAvg$, all the other functions only returns boolean values 
(0 or~1).
Given an condition $\Phi: \Path \to \reals$, 
a probabilistic automaton $\cala$ and a word $w$, we denote 
by $\Exp^w_\cala(\Phi)$ the expectation of the function $\Phi$ 
under the probability measure $\Prb^w_\cala$.
Given a probabilistic automaton $\cala$ and a condition $\Phi$, 
we use the following notation: $\cala(\Phi,w)=\Exp^w_\cala(\Phi)$,
and if the condition $\Phi$ is clear from the context we simply
write $\cala(w)$.
If $\Phi$ is boolean, then $\cala(\Phi,w)$ is the acceptance probability 
for the word $w$ for the condition $\Phi$ in $\cala$.

\subsection{Decision problems}
We now consider the quantitative and qualitative decision problems.

\smallskip\noindent{\bf Quantitative decision problems.}
Given a probabilistic automaton $\cala$, a condition $\Phi$, 
and a rational number $0< \lambda < 1$, we consider the following 
questions:
\begin{enumerate}

\item {\em Quantitative equality problem.} 
Does there exist a word $w$ such that $\Exp^w_\cala(\Phi)=\lambda$.
If $\Phi$ is boolean, then the question is whether there exists a 
word $w$ such that $\Prb^w_\cala(\Phi)=\lambda$.

\item {\em Quantitative existence problem.} 
Does there exist a word $w$ such that $\Exp^w_\cala(\Phi) > \lambda$.
If $\Phi$ is boolean, then the question is whether there exists a 
word $w$ such that $\Prb^w_\cala(\Phi) > \lambda$.
This question is related to emptiness of probabilistic automata:
let $\call_\cala(\Phi,\lambda)=\set{w \mid \Prb_\cala^w(\Phi)> \lambda}$
be the set of words with acceptance probability greater than $\lambda$;
then the set is non-empty iff the answer to the quantitative
existence problem is yes.

\item {\em Quantitative value problem.}
Whether the supremum of the values for all words is greater than 
$\lambda$, i.e., $\sup_{w \in \alphab^\omega} \Exp^w_\cala(\Phi) > \lambda$.
If $\Phi$ is boolean, this question is equivalent to whether for all
$\vare>0$, does there exist a word $w$ such that $\Prb^w_\cala(\Phi)
> \lambda -\vare$.


\end{enumerate}

\smallskip\noindent{\bf Qualitative decision problems.}
Given a probabilistic automaton $\cala$, a condition $\Phi$, 
we consider the following questions:
\begin{enumerate}
\item {\em Almost problem.} Does there exist a word $w$ such 
that $\Exp_\cala^w(\Phi)=1$.

\item {\em Positive problem.} Does there exist a word $w$ such 
that $\Exp_\cala^w(\Phi)>0$.

\item {\em Limit problem.} For all $\vare>0$, does there exist 
a word $w$ such that $\Exp_\cala^w(\Phi)>1-\vare$.
\end{enumerate}
If $\Phi$ is boolean, then in all the above questions $\Exp$ is replaced
by $\Prb$.

\section{Undecidability of Quantitative Decision Problems}
In this section we study the computability of the quantitative decision 
problems. We show undecidability results for special classes of 
probabilistic automata for safety and reachability conditions, and all
other results are derived from the results on these special classes.
The special classes are related to the notion of absorption in probabilistic
automata.

\subsection{Absorption in probabilistic automata}
We will consider several special cases with absorption condition 
and consider some simple equivalences.

\smallskip\noindent{\bf Absorbing states.}
Given a probabilistic automaton $\cala$, a state $q$ is 
\emph{absorbing} if for all $\sigma \in \alphab$ we 
have $M_\sigma(q,q)=1$ (hence for all $q' \neq q$ we 
have $M_\sigma(q,q')=0$).

\smallskip\noindent{\bf Acceptance absorbing automata.}
Given a probabilistic automaton $\cala$ let $F$ be the set of 
accepting states. 
The automaton is \emph{acceptance-absorbing} if all states in 
$F$ are absorbing.
Given an acceptance-absorbing automaton, the following equivalence hold:
$\Reach(F)=\Buchi(F) =\coBuchi(F)=\LimitAvg(F)$.
Hence our goal would be to show hardness (undecidability) 
for acceptance-absorbing automata with reachability condition, 
and the results will follow for B\"uchi, coB\"uchi and 
limit-average conditions.

\smallskip\noindent{\bf Absorbing automata.}
A probabilistic automaton $\cala$ is absorbing if the following 
condition holds: 
let $C$ be the set of absorbing states in $\cala$, then for 
all $\sigma \in \alphab$ and for all $q \in (Q \setminus C)$ 
we have $M_\sigma(q,q')>0$ for some $q'\in C$.
In sequel we will use $C$ for absorbing states in a probabilistic
automaton.

\subsection{Absorbing safety automata}
In sequel we write automata (resp. automaton) to denote 
probabilistic automata (resp. probabilistic automaton) unless
mentioned otherwise.
We now prove some simple properties of 
absorbing automata with safety condition.

\begin{lemma}
Let $\cala$ be an absorbing automaton with $C$ as the 
set of absorbing states.
Then for all words $w \in \alphab^\omega$ we 
have $\Prb^w_\cala(\Reach(C))=1$.
\end{lemma}
\begin{proof}
Let $\eta= \min_{q \in Q, \sigma \in \alphab} \sum_{q' \in C}
M_\sigma(q,q')$ be the minimum transition probability to  
absorbing states.
Since $\cala$ is absorbing we have $\eta>0$.
Hence  for any word $w$, the probability to 
reach $C$ after $n$ steps is at least $1-(1-\eta)^n$.
Since $\eta>0$, as $n \to \infty$ we have $1-(1-\eta)^n \to 1$,
and hence the result follows.
\qed
\end{proof}

\smallskip\noindent{\bf Complementation of absorbing safety automata.}
Let $\cala$ be an absorbing automaton, with the set $F$ as accepting states,
and we consider the safety condition $\Safe(F)$. 
Without loss of generality we assume that every state in 
$Q \setminus F$ is absorbing.
Otherwise, if a state in $q \in Q \setminus F$ is non-absorbing,
we transform it to an absorbing state and obtain an automaton $\cala'$.
It is easy to show that for the safety condition $\Safe(F)$ the automata 
$\cala$ and $\cala'$ coincide 
(i.e., for all words $w$ we have $\cala(\Safe(F),w)=\cala'(\Safe(F),w)$).
Hence we consider an absorbing safety automaton such that all states in 
$Q\setminus F$ are absorbing: it follows that every non-absorbing state
is an accepting state, i.e., $Q \setminus C \subseteq F$.
We claim that for all words $w$ we have 
\[
\Prb^w_\cala(\Safe(F))= \Prb^w_\cala(\Reach(F \cap C)).
\]
Since every state in $Q \setminus C \subseteq F$ and all states 
in $C$ are absorbing, it follows that 
$\Prb^w_\cala(\Safe(F)) \geq \Prb^w_\cala(\Reach(F \cap C))$.
Since $\cala$ is absorbing, for all words $w$ we have 
$\Prb^w_\cala(\Reach(C))=1$ and hence it follows that 
\[
\Prb^w_\cala(\Safe(F)) \leq \Prb^w_\cala(\Reach(C)) \cdot \Prb^w_\cala(\Safe(F)\mid \Reach(C)) 
=\Prb^w_\cala(\Reach(F \cap C)).
\]
The complement automaton $\ov{\cala}$ is obtained from $\cala$ by changing the 
set of accepting states as follows: 
the set of accepting states $\ov{F}$ in $\ov{\cala}$ is 
$(Q \setminus C) \cup (C \setminus F)$, i.e., all non-absorbing states 
are accepting states, and for absorbing states the accepting states are switched.
Since $\ov{\cala}$ is also absorbing it follows that for all words $w$ we 
have  
$\Prb^w_{\ov{\cala}}(\Safe(\ov{F})) =\Prb^w_{\ov{\cala}}(\Reach(C \cap \ov{F})) 
= \Prb^w_{\ov{\cala}}(\Reach(C \setminus F))$.
Since $\ov{\cala}$ is absorbing, it follows that for all words we have 
$\Prb^w_{\ov{\cala}}(\Reach(C))=1$ and hence 
$\Prb^w_{\ov{\cala}}(\Safe(\ov{F}))=\Prb^w_{\ov{\cala}}(\Reach(C \setminus F)) = 
1-\Prb^w_{\ov{\cala}}(\Reach(C\cap F))= 1-\Prb^w_{\cala}(\Reach(C\cap F))$.
It follows that for all words $w$ we have 
$\cala(\Safe(F),w) =1-\ov{\cala}(\Safe(\ov{F}),w)$, i.e., 
$\ov{\cala}$ is the complement of $\cala$.

\begin{remark}
It follows from above that an absorbing automaton with safety 
condition can be transformed to an acceptance-absorbing automaton
with reachability condition. 
So any hardness result for absorbing safety automata also gives 
a hardness result for acceptance-absorbing reachability automata 
(and hence also for B\"uchi, coB\"uchi, limit-average automata).
\end{remark}

\subsection{Undecidability for safety and acceptance-absorbing reachability automata}

Our first goal is to show that the quantitative equality problem is 
undecidable for safety and acceptance-absorbing reachability automata.
The reduction is from the Post Correspondence Problem (PCP). 
Our proof is inspired by and is an extension of a simple elegant proof 
of undecidability for probabilistic finite automata~\cite{GO-Tech}.
We first define the PCP and some related notations.

\smallskip\noindent{\bf PCP.} 
Let $\varphi_1,\varphi_2: \alphab \to \set{0,1,\ldots,k-1}^*$, and extended naturally 
to $\alphab^*$ (where $k=|\alphab|$).
The PCP  asks whether there is a word $w \in \alphab^*$ such that 
$\varphi_1(w)=\varphi_2(w)$.
The PCP is undecidable~\cite{HopcroftUllman}.

\smallskip\noindent{\em Notations.}
Let $\psi: \set{0,1,\ldots,k-1}^* \to [0,1]$ be the function defined as:
\[
\psi(\sigma_1 \sigma_2 \ldots \sigma_n) = 
\frac{\sigma_n}{k} + 
\frac{\sigma_{n-1}}{k^2} + 
 \cdots 
+ \frac{\sigma_2}{k^{n-1}}
+ \frac{\sigma_1}{k^n} \ .
\]
For $i \in \set{1,2}$, let $\theta_i=\psi \circ \varphi_i: \Sigma^* \to [0,1]$.
We first prove a property of $\theta_i$, that can be derived from the 
results of~\cite{Bertoni}.

\begin{lemma}\label{lemm:bert}
For a finite word $w \in \alphab^*$ and a letter $\sigma \in \alphab$ we have 
\[
\theta_i(w \cdot \sigma)= 
\theta_i(\sigma) + \theta_i(w) \cdot k_i(\sigma), 
\]
where $k_i(\sigma)=k^{-|\varphi_i(\sigma)|}$.
\end{lemma}
\begin{proof}
Let $w=\sigma_1 \sigma_2 \ldots \sigma_n$, and let 
$\varphi_i(w)=b_1 b_2 \ldots b_m$ 
and $\varphi_i(\sigma)= a_1 a_2 \ldots a_\ell$.
Then we have 
\[
\begin{array}{rcl}
\theta_i(w \cdot \sigma) & = & \psi \circ \varphi_i(w \cdot \sigma) \\[1ex]
 & = & 
\displaystyle 
\psi ( b_1 b_2 \ldots b_m a_1 a_2 \ldots a_\ell) \\[1ex]
 & = & 
\displaystyle 
\frac{a_\ell}{k} + \cdots + \frac{a_1}{k^\ell} + 
\frac{b_m}{k^{\ell+1}} + \cdots + \frac{b_1}{k^{\ell+m}} \\[1ex]
 & = & 
\displaystyle 
\bigg(\frac{a_\ell}{k} + \cdots + \frac{a_1}{k^\ell}\bigg) + 
\frac{1}{k^\ell} \cdot \bigg(\frac{b_m}{k} + \cdots + \frac{b_1}{k^{m}} \bigg) \\[1ex]
& = & 
\displaystyle 
\psi \circ \varphi_i(\sigma) + \frac{1}{k^\ell} \cdot( \psi\circ \varphi_i(w)) \\[1ex]
& = & \theta_i(\sigma) + \theta_i(w) \cdot k_i(\sigma).
\end{array} 
\]
The result follows.
\qed
\end{proof}

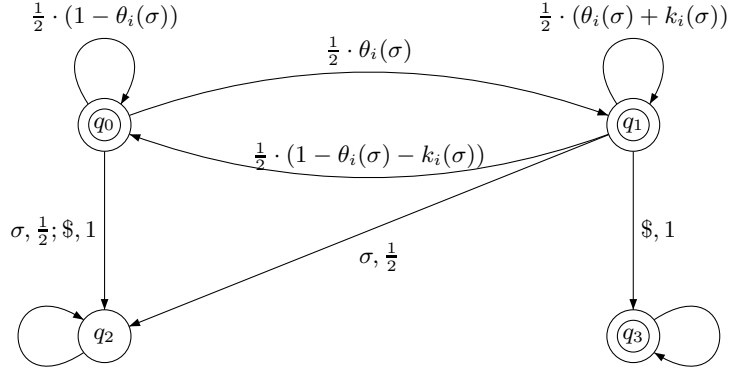
\begin{figure}[!tb]
  \begin{center}
    \unitlength=4pt
    \begin{picture}(25,25)(0,-20)
    \gasset{Nw=5,Nh=5,Nmr=2.5,curvedepth=5}
    \thinlines
    \node[Nmarks=i,iangle=90,Nmarks=r,rdist=1](A1)(-10,0){$q_0$}
    \drawloop[loopangle=90](A1){$\frac{1}{2}\cdot(1-\theta_i(\sigma))$}
    \node[Nmarks=r,rdist=1](A2)(40,0){$q_1$}
    \drawloop[loopangle=90](A2){$\frac{1}{2}\cdot(\theta_i(\sigma)+k_i(\sigma))$}
    \node(A3)(-10,-20){$q_2$}
    \node[Nmarks=r,rdist=1](A4)(40,-20){$q_3$}
    \drawloop[loopangle=180](A3){}
    \drawloop[loopangle=0](A4){}
    \drawedge(A1,A2){$\frac{1}{2}\cdot \theta_i(\sigma)$}
    \drawedge[ELside=r](A2,A1){$\frac{1}{2}\cdot(1-\theta_i(\sigma)-k_i(\sigma))$}
    \gasset{curvedepth=0}
    \drawedge(A2,A3){$\sigma,\frac{1}{2}$}
    \drawedge(A2,A4){$\$,1$}
    \drawedge[ELside=r](A1,A3){$\sigma,\frac{1}{2}; \$,1$}
    \end{picture}
  \end{center}
  \caption{Absorbing safety automata from PCP instances.}\label{fig:safe}
\end{figure}

\smallskip\noindent{\bf The absorbing safety automata from PCP instances.} 
Given an instance of the PCP problem we create two safety automata 
$\cala_i=(Q,\alphab \cup \set{\$},\calm^i,q_\ii,F)$, 
for $i \in \set{1,2}$ as follows (we assume $\$\not\in \alphab$): 
\begin{enumerate}
\item \emph{(Set of states and initial state).} $Q=\set{q_0,q_1,q_2,q_3}$ and $q_\ii=q_0$;
\item \emph{(Accepting states).} $F=\set{q_0,q_1,q_3}$;
\item \emph{(Transition matrices).} The set of transition matrices is as follows:
\begin{enumerate}
\item the states $q_2$ and $q_3$ are absorbing;
\item $M^i_{\$}(q_0,q_2)=1$ and $M^i_{\$}(q_1,q_3)=1$; 
\item for all $\sigma \in \alphab$ we have 
\begin{enumerate}
\item $M^i_\sigma(q_0,q_2)=M^i_\sigma(q_1,q_2)=\frac{1}{2}$;
\item $M^i_\sigma(q_0,q_0)= \frac{1}{2}\cdot(1-\theta_i(\sigma))$; 
$M^i_\sigma(q_0,q_1)=\frac{1}{2}\cdot \theta_i(\sigma)$;
\item $M^i_\sigma(q_1,q_0)= \frac{1}{2}\cdot(1-\theta_i(\sigma)-k_i(\sigma))$; 
$M^i_\sigma(q_1,q_1)=\frac{1}{2}\cdot(\theta_i(\sigma)+k_i(\sigma))$;
\end{enumerate}
\end{enumerate} 
\end{enumerate}
A pictorial description is shown in Fig~\ref{fig:safe}.
We will use the following notations:
(a)~we use $\wh{\alphab}$ for $\alphab \cup \set{\$}$;
(b)~for a word $w \in \wh{\alphab}^\omega$ if the word contains a $\$$,
then we denote by $\first(w)$ the prefix $w'$ of $w$ that does not 
contain a $\$$ and the first $\$$ in $w$ occurs immediately after $w'$ 
(i.e., $w'\$$ is a prefix of $w$).
We now prove some properties of the automata $\cala_i$.
\begin{enumerate}
\item The automata $\cala_i$ is absorbing: since for every state and 
every letter the transition probability to the set $\set{q_2,q_3}$ 
is at least $\frac{1}{2}$; and $q_2$ and $q_3$ are absorbing.

\item Consider a word $w \in \wh{\alphab}^\omega$. 
If the word contains no $\$$, then the state $q_2$ is reached with 
probability~1 (as every round there is at least probability $\frac{1}{2}$ 
to reach $q_2$). 
If the word $w$ contains a $\$$, then as soon as the input letter is 
$\$$, then the set $\set{q_2,q_3}$ is reached with probability~1.
Hence the following assertion holds: the probability 
$\Prb^w_\cala(\Safe(F))$ is the probability that after the word $w'=\first(w)$ 
the current state is $q_1$. 
\end{enumerate}

\begin{lemma}\label{lemm:theta}
For all words $w \in \alphab^*$, the probability that in the 
automaton $\cala_i$ after reading $w$ 
(a)~the current state is $q_1$ is equal to $\frac{1}{2^{|w|}} \cdot \theta_i(w)$;
(b)~the current state is $q_0$ is equal to $\frac{1}{2^{|w|}} \cdot (1-\theta_i(w))$;
and (c)~the current state is $q_2$ is equal to $1-\frac{1}{2^{|w|}}$.
\end{lemma}
\begin{proof} The result follows from induction on length of $w$, and 
the base case is trivial.
We prove the inductive case.
Consider a word $w\cdot \sigma$: by inductive hypothesis the probability that 
after reading $w$ the current state is 
$q_0$ is $\frac{1}{2^{|w|}}\cdot (1-\theta_i(w))$ and the current state is 
$q_1$ is $\frac{1}{2^{|w|}}\cdot \theta_i(w)$, and the current state is 
$q_2$ with probability $(1-\frac{1}{2^{|w|}})$.
After reading $\sigma$, if the current state is $q_0$ or $q_1$, 
with probability $\frac{1}{2}$ a transition to $q_2$ is made. 
Hence the probability to be at $q_2$ after $w\cdot \sigma$ is 
$(1-\frac{1}{2^{|w|+1}})$ and the rest of the probability is to be at either 
$q_0$ and $q_1$.
The probability to be at $q_1$ is 
\[
\begin{array}{rcl}
\frac{1}{2^{|w|+1}} \cdot \big(
(1-\theta_i(w))\cdot \theta_i(\sigma) 
+
\theta_i(w)\cdot (\theta_i(\sigma) + k_i(\sigma)
\big)
& = & 
\frac{1}{2^{|w|+1}} \cdot (\theta_i(\sigma) + \theta_i(w)\cdot k_i(\sigma)) \\
& = &
\frac{1}{2^{|w|+1}} \cdot \theta_i(w\cdot \sigma).
\end{array}
\]
The first equality follows by rearranging and the 
second equality follows from Lemma~\ref{lemm:bert}.
Hence the result follows.
\qed
\end{proof}

The following lemma is an easy consequence.

\begin{lemma}\label{lemm:theta_val}
For $i \in \set{1,2}$, for a word $w \in \wh{\alphab}^\omega$, 
(a)~if $w$ contains no $\$$, then $\cala_i(w)=0$;
(b)~if $w$ contains a $\$$,  let $w'=\first(w)$, then 
$\cala_i(w)=\frac{1}{2^{|w'|}}\cdot \theta_i(w')$. 
\end{lemma}

\smallskip\noindent{\bf Constant automata and random choice automata.}
For any rational constant $\nu$ it is easy to construct an absorbing safety 
automaton $\cala$ that assigns value $\nu$ to all words.
Given two absorbing safety automata $\cala_1$ and $\cala_2$, and 
two non-negative rational numbers $\beta_1,\beta_2$ such that 
$\beta_1 + \beta_2=1$, it is easy to construct an automaton $\cala$ 
(by adding an initial state with initial randomization) 
such that for all words $w$ we have 
$\cala(w)=\beta_1\cdot\cala_1(w) + \beta_2\cdot\cala_2(w)$.
We will use the notation $\beta_1\cdot \cala_1 + \beta_2\cdot\cala_2$
for the automaton $\cala$.

\begin{figure}[!tb]
  \begin{center}
    \unitlength=4pt
    \begin{picture}(25, 10)(0,-5)
    \gasset{Nw=5,Nh=5,Nmr=2.5,curvedepth=0}
    \thinlines
    \node[Nmarks=i,iangle=90,Nmarks=r,rdist=1](A1)(0,0){$q_0$}
    \drawloop[loopangle=180](A1){$\alphab$}
    \node(A2)(25,0){$q_1$}
    \drawloop[loopangle=0](A2){$\wh{\alphab}$}
    \drawedge(A1,A2){$\$$}
    \end{picture}
  \end{center}
  \caption{Safety automaton $\cala_3$.}\label{fig:aut3}
\end{figure}
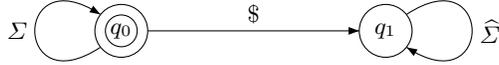

\smallskip\noindent{\bf Safety automaton $\cala_3$.} 
We consider a safety automaton $\cala_3$ on the alphabet 
$\wh{\alphab}$ as follows: 
(a)~for a word $w$ without any $\$$ the acceptance probability
is~1;
(b)~for a word with a $\$$ the acceptance probability is~0.
The automaton is shown in Fig~\ref{fig:aut3} and the only accepting 
state is $q_0$.
Consider the acceptance-absorbing reachability automaton $\cala_4$ 
as follows: the automaton is same as in Fig~\ref{fig:aut3} with 
accepting state as $q_1$: the automaton accepts a word with a $\$$ 
with probability~1, and for words with no $\$$ the acceptance 
probability is~0.

\begin{theorem}[Quantitative equality problem]\label{thrm1}
The quantitative equality problem is undecidable for probabilistic safety and 
acceptance-absorbing reachability automata.
\end{theorem}
\begin{proof}
Consider an automaton 
$\cala= \frac{1}{3}\cdot\cala_1 + \frac{1}{3}\cdot(1-\cala_2) + \frac{1}{3}\cdot \cala_3$
(since $\cala_2$ is absorbing we can complement $\cala_2$ and obtain 
an automaton for $1-\cala_2$).
We show that the quantitative equality problem for $\cala$ with $\lambda=\frac{1}{3}$ 
is yes iff the answer to the PCP problem instance is yes.
We prove the following two cases.
\begin{enumerate}
\item Suppose there is a finite word $w$ such that 
$\varphi_1(w)=\varphi_2(w)$. 
Consider the infinite word $w^*=w \$ w'$ where $w'$ is an 
arbitrary infinite word.
Then the acceptance probability of $w^*$ in $\cala_3$ is~0 
(since it contains a $\$$).
Since $w^*$ contains a $\$$ and $\first(w^*)=w$,
by Lemma~\ref{lemm:theta_val} the acceptance probability of $w^*$ in 
$\cala$ is $\frac{1}{3} + \frac{1}{2^{|w|}}\cdot(\theta_1(w)-\theta_2(w))$.
Since $\varphi_1(w)=\varphi_2(w)$ it follows that
$\theta_1(w)-\theta_2(w)=0$.
It follows that the acceptance probability of $w^*$ in $\cala$ is 
$\frac{1}{3}$.

\item Consider an infinite word $w$.
If the word contains no $\$$, then the acceptance probability of $w$ in
$\cala_1$ and $\cala_2$ is~0, and the 
acceptance probability in $\cala_3$ is~1.
Hence $\cala$ accepts $w$ with probability $\frac{2}{3} > \frac{1}{3}$.
If the word $w$ contains a $\$$, then $\cala_3$ accepts 
$w$ with probability~0.
The difference of acceptance probability of $w$ in $\cala_1$ and $\cala_2$ 
is 
$\frac{1}{2^{|w'|}}\cdot(\theta_1(w')-\theta_2(w'))$, where $w'=\first(w)$.
Hence the acceptance probability of $w$ in $\cala$ is $\frac{1}{3}$ iff 
$\theta_1(w')=\theta_2(w')$.
Hence $w'$ is a witness that $\varphi_1(w')=\varphi_2(w')$.
\end{enumerate}
It follows that there exists a finite word $w$ that is a witness
to the PCP instance iff there exists an infinite word $w^*$ in 
$\cala$ with acceptance probability equal to $\frac{1}{3}$.

For acceptance-absorbing reachability automata: consider the same 
construction as above with $\cala_3$ being replaced by $\cala_4$, 
and the equality question for $\lambda=\frac{2}{3}$.
Since $\cala_1-\cala_2 < 1$, it follows that any witness word must 
contain a $\$$, and then the proof is similar as above. The result for acceptance-absorbing 
reachability automata also follows easily from \cite{Bertoni}.\qed
\end{proof}

\smallskip\noindent{\bf Quantitative existence and value problems.}
We will now show that the quantitative existence and the quantitative value 
problems are undecidable for probabilistic absorbing safety automata.
We start with a technical lemma.

\begin{lemma}\label{lemm:tech}
Let us consider a PCP instance.
Let $z=\max_{\sigma \in \alphab} \set{|\varphi_1(\sigma)|, |\varphi_2(\sigma)|}$.
For a finite word $w$, if 
$\theta_1(w) - \theta_2(w) \neq 0$, then 
$(\theta_1(w) - \theta_2(w) )^2 \geq \frac{1}{k^{2\cdot |w|\cdot z }}$.
\end{lemma}
\begin{proof}
Given the word $w$, we have 
$|\varphi_i(w)| \leq |w|\cdot z$, for $i \in \set{1,2}$.
It follows that $\theta_i(w)$ can be expressed as a rational number
$\frac{p_i}{q}$, where $q \leq k^{z\cdot|w|}$. 
It follows that if $\theta_1(w) \neq \theta_2(w)$, then 
$|\theta_1(w) -\theta_2(w)| \geq \frac{1}{k^{z \cdot |w|}}$.
The result follows.
\qed
\end{proof}

\begin{figure}[!tb]
  \begin{center}
    \unitlength=4pt
    \begin{picture}(25, 15)(-10,-5)
    \gasset{Nw=5,Nh=5,Nmr=2.5,curvedepth=0}
    \thinlines
    \node[Nmarks=i,iangle=270,Nmarks=r,rdist=1](A1)(0,0){$q_0$}
    \drawloop[loopangle=90](A1){$\alphab, \frac{1}{k^{2\cdot(z+1)}}$}
    \node[Nmarks=r,rdist=1](A2)(25,0){$q_1$}
    \drawloop[loopangle=0](A2){$\wh{\alphab}$}
    \drawedge(A1,A2){$\$$}
    \node(A3)(-25,0){$q_2$}
    \drawloop[loopangle=180](A3){$\wh{\alphab}$}
    \drawedge(A1,A3){$\alphab,1- \frac{1}{k^{2\cdot(z+1)}}$}
    \end{picture}
  \end{center}
  \caption{Safety automaton $\cala_5$.}\label{fig:aut5}
\end{figure}
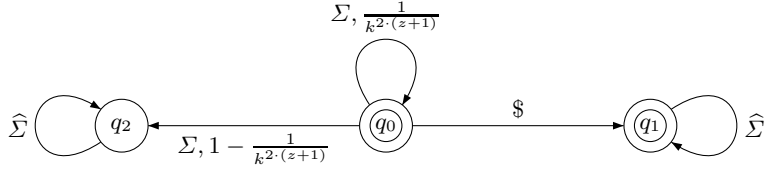

\smallskip\noindent{\bf Automaton $\cala_5$.}
Consider the automaton shown in Fig~\ref{fig:aut5}: 
the accepting states are $q_0$ and $q_1$. 
For any input letter in $\alphab$, from the initial state $q_0$ 
the next state is itself with probability $\frac{1}{k^{2\cdot(z+1)}}$,
and the rest of the probability is to goto $q_2$.
For input letter $\$$ the next state is $q_1$ with probability~1.
The states $q_1$ and $q_2$ are absorbing.
Given a word $w$, if there is no $\$$ in $w$, then the acceptance probability is~0;
otherwise the acceptance probability is 
$\frac{1}{k^{2(z+1)\cdot|w'|}}$, where $w'=\first(w)$.
Also note that the automaton $\cala_5$ is absorbing.

\begin{theorem}[Quantitative existence and value problems]\label{thrm2}
The quantitative existence and value problems are undecidable for 
probabilistic absorbing safety automata.
\end{theorem}
\begin{proof}
Let us first consider the following automata: 
$\calb_1= \frac{1}{2}\cdot \cala_1 + \frac{1}{2}\cdot (1-\cala_2)$ 
and 
$\calb_2= \frac{1}{2}\cdot \cala_2 + \frac{1}{2}\cdot (1-\cala_1)$. 
The Cartesian product of $\calb_1$ and $\calb_2$ 
gives us the automaton $\calb_3= \frac{1}{4} - (\cala_1 -\cala_2)^2$.
Finally, let us consider the automaton
$\calb=\frac{1}{2}\cdot \calb_3 + \frac{1}{2} \cdot \cala_5$.
The automaton $\calb$ can be obtained as an absorbing safety automaton.
We show that there exists a word (or the sup value of words) is 
greater than $\frac{1}{8}$ iff the answer to PCP is yes.

\begin{enumerate}
\item If the answer to PCP problem is yes, consider the 
witness finite word $w$ such that $\varphi_1(w)=\varphi_2(w)$.
We construct an infinite word $w^*$ for $\calb$ as follows: 
$w^*=w \$ w'$ where $w'$ is an arbitrary infinite word.
Since the word $w^*$ contains a $\$$ by Lemma~\ref{lemm:theta_val} we have 
the difference in acceptance probability of $w^*$ in 
$\cala_1$ and $\cala_2$ is $\frac{1}{2^{|w|}}\cdot (\theta_1(w)-\theta_2(w))$.
Since $\varphi_1(w)=\varphi_2(w)$, it follows that 
$\cala_1(w^*)=\cala_2(w^*)$.
Hence we have 
$\calb(w^*)= \frac{1}{8} + \frac{1}{{k^{2 \cdot(z+1)\cdot|w|}}} > \frac{1}{8}$.

\item We now show that if there is an infinite word with 
acceptance probability greater than $\frac{1}{8}$ (or the sup over the infinite 
words of the acceptance probability is greater than $\frac{1}{8}$)
in $\calb$, then the answer to the PCP problem is yes.  
Consider an infinite word $w$ for $\calb$.
If $w$ contains no $\$$, then $\cala_1$, $\cala_2$, and 
$\cala_5$ accepts $w$ with probability~0, and hence 
the acceptance probability in $\calb$ is $\frac{1}{8}$.
Consider an infinite word $w$ that contains a $\$$.
Let $w'=\first(w)$.
If $\theta_1(w') \neq \theta_2(w')$, then by Lemma~\ref{lemm:theta_val} 
and Lemma~\ref{lemm:tech} we have 
$(\cala_1(w)-\cala_2(w))^2 \geq 
\frac{1}{2^{2|w'|}} \cdot 
\frac{1}{k^{2\cdot |w'|\cdot z }} 
\geq \frac{1}{k^{2\cdot |w'|\cdot (z+1) }}$.
Since $\cala_5(w)=  \frac{1}{k^{2\cdot |w'|\cdot (z+1) }}$,
it follows that $\calb(w) \leq 1/8$. 
If $\theta_1(w')=\theta_2(w')$ (which implies $\varphi_1(w')=\varphi_2(w')$), 
then 
$\calb(w) =\frac{1}{8} +  \frac{1}{k^{2\cdot |w'|\cdot (z+1) }} > \frac{1}{8}$.
\end{enumerate}
It follows from above that the quantitative existence and 
the quantitative value problems are undecidable for 
absorbing safety (and hence also for acceptance-absorbing 
reachability) automata.
\qed
\end{proof}

\begin{corollary}
Given any rational $0<\lambda<1$, the quantitative equality, 
existence, and 
value problems are undecidable for probabilistic 
absorbing safety and acceptance-absorbing reachability automata.  
\end{corollary}
\begin{proof}
We have shown in Theorem~\ref{thrm1} and Theorem~\ref{thrm2} that the problems 
are undecidable for specific constants (e.g., $\frac{1}{3}, \frac{1}{8}$ etc.).
We show below given the problem is undecidable for a constant $0< c < 1$,
the problem is also undecidable for any given rational $0< \lambda < 1$.
Given an automaton $\calb$ we consider two cases:
\begin{enumerate}
\item If $\lambda \leq c$, then consider the 
automaton $\cala= \frac{\lambda}{c} \cdot \calb + (1-\frac{\lambda}{c})\cdot 0$.
The quantitative decision problems for constant $c$ in $\calb$ is exactly same 
as the quantitative decision problems for $\cala$ with $\lambda$.

\item If $\lambda \geq c$, then consider the 
automaton $\cala=\frac{1-\lambda}{1-c}\cdot \calb + \frac{\lambda-c}{1-c}\cdot 1$.
The quantitative decision problems for $c$ in $\calb$ is exactly same 
as the quantitative decision problems for $\cala$ with $\lambda$.
\end{enumerate}
The desired result follows. The result for acceptance-absorbing reachability 
automata also follows from \cite{Bertoni}.
\qed
\end{proof}

\begin{corollary}[Quantitative decision problems]
Given any rational $0<\lambda<1$, the quantitative equality, 
existence and 
value problems are undecidable for probabilistic safety, 
reachability, B\"uchi, coB\"uchi and limit-average automata.
\end{corollary}

\section{(Un-)Decidability of Qualitative Decision Problems}
In this section we show that the qualitative decision problems 
are decidable for safety automata, and the limit problem is 
undecidable for reachability, B\"uchi, coB\"uchi and limit-average 
automata.
The result is inspired from the proof of~\cite{GO-Tech} that shows 
that the limit problem is undecidable for finite automata.
If the finite automata construction of~\cite{GO-Tech} were acceptance 
absorbing, the result would have followed for acceptance-absorbing 
reachability automata (and hence also for B\"uchi, coB\"uchi and 
limit-average automata).
However, the undecidability proof of~\cite{GO-Tech} for finite automata 
constructs automata that are not acceptance-absorbing. 
Simply changing the accepting states of the automata constructed 
in~\cite{GO-Tech} to absorbing states does not yield the undecidability for 
accepting absorbing automata.
We show that the construction of~\cite{GO-Tech} can be adapted to 
prove undecidability of the limit problem for acceptance-absorbing 
automata with reachability condition.
We first present a simple proof that for safety condition the almost and limit
problem coincide.

\begin{lemma}
Given an automaton with a safety condition the answer to the limit problem
is yes iff the answer to the almost problem is yes.
\end{lemma}
\begin{proof}
If the answer to the almost problem is yes, then trivially the answer 
to the limit problem is yes (since a witness for almost problem is a 
witness for all $\vare>0$ for the limit problem).
We now show the converse is true.
Consider an automaton $\cala$ with $\ell$ states, and let $\eta>0$ be the 
minimum non-zero transition probability in $\cala$.
We assume that the answer to the limit problem is yes and show that the 
answer to the almost problem is also yes.
Consider $\vare=\eta^{2^\ell}$, and let $w$ be a word such that 
$\cala(w) > 1-\vare$.
For a position $i$ of $w$ let $S_i=\set{q \in Q \mid \delta_{|w_i|+1}(q) >0}$ 
be the set of states that have positive probability in the distribution of 
states after reading the prefix of length $i$ of $w$.
If a path in $\cala$ leaves the set $F$ of accepting states within $k$ steps, 
then it leaves with probability at least $\eta^k$. 
Since $w$ ensures the safety condition with probability at least 
$1-\vare =1 - \eta^{2^\ell}$ it follows that for all $1\leq i \leq 2^\ell$ we have 
$S_i \subseteq F$ (i.e., each $S_i$ upto $2^\ell$ is a subset of the accepting 
states).
There must exist  $0\leq n < m \leq 2^\ell$ such that $S_n=S_m$.
Consider the word $w^* = w_1 w_2 \cdots w_{n-1} (w_n w_{n+1} w_{m-1})^\omega$.
The word $w^*$ ensures the following: let 
$S^*_i=\set{q \in Q \mid \delta_{|w^*_i|+1}(q) >0}$, then 
for all $0\leq i \leq n$ we have $S_i=S^*_i$, and for all $i \geq n$ we have 
$S^*_i \subseteq \bigcup_{j=n}^m S_i$.
It follows that for all $i \geq 0$ we have $S^*_i \subseteq F$ and hence $\cala(w^*)=1$. 
The result follows.
\qed
\end{proof}

\smallskip\noindent{\bf Qualitative decision problems for safety conditions.}
The above lemma shows that for safety condition the answers to the  
almost and the limit problem coincide.
The almost and positive problems for safety condition are 
PSPACE-complete~\cite{chadha2009expressiveness}.
It follows that the qualitative decision problems are decidable for safety 
condition.

\begin{figure}[!tb]
  \begin{center}
    \unitlength=4pt
    \begin{picture}(25,45)(20,-25)
    \gasset{Nw=5,Nh=5,Nmr=2.5,curvedepth=4}
    \thinlines
    \node[Nmarks=i,iangle=180](A1)(0,0){$q_1$}
    \drawloop[loopangle=90](A1){$a,x$}
    \node(A2)(20,0){$q_2$}
    \drawloop[loopangle=90](A2){$a$}
    \node[Nmarks=r,rdist=1](A3)(0,-16){$q_3$}
    \node(A4)(20,-16){$q_4$}
    \drawloop[loopangle=270](A3){$a,b,\$$}
    \drawloop[loopangle=270](A4){$a,b,\$$}
    \drawedge(A1,A2){$a,1-x$}
    \drawedge[ELside=r](A2,A1){$b$}
    \gasset{curvedepth=0}
    \drawedge(A1,A4){$\$$}
    \drawedge(A2,A4){$\$$}
    \drawedge(A1,A3){$b$}

    \gasset{Nw=5,Nh=5,Nmr=2.5,curvedepth=4}
    \thinlines
    \node[Nmarks=i,iangle=180](A5)(40,0){$q_5$}
    \drawloop[loopangle=90](A5){$a,1-x$}
    \node(A6)(60,0){$q_6$}
    \drawloop[loopangle=90](A6){$a$}
    \node(A7)(40,-16){$q_7$}
    \node[Nmarks=r,rdist=1](A8)(60,-16){$q_8$}
    \drawloop[loopangle=270](A7){$a,b,\$$}
    \drawloop[loopangle=270](A8){$a,b,\$$}
    \drawedge(A5,A6){$a,x$}
    \drawedge[ELside=r](A6,A5){$b$}
    \gasset{curvedepth=0}
    \drawedge(A5,A8){$\$$}
    \drawedge(A6,A8){$\$$}
    \drawedge(A5,A7){$b$}

    \end{picture}
  \end{center}
  \caption{Automata $\cala_1$ (on left) and $\cala_2$ (on right).}\label{fig:limit}
\end{figure}
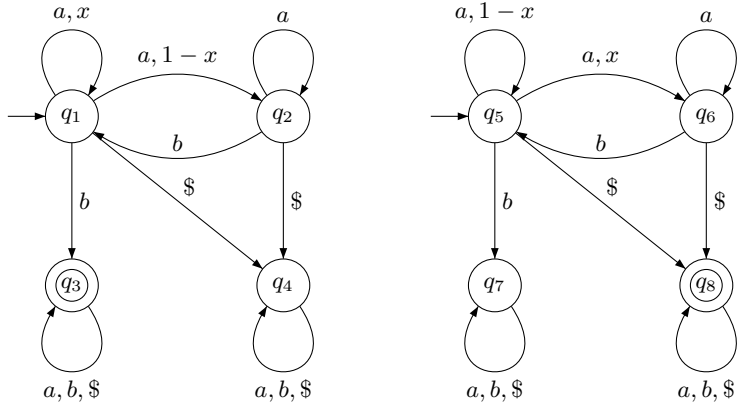

\begin{lemma}\label{lemm:limit_lemma}
Consider the acceptance-absorbing automaton $\cala=
\frac{1}{2}\cala_1 + \frac{1}{2}\cdot \cala_2$, where $\cala_1$ and 
$\cala_2$ are shown in Fig~\ref{fig:limit}.
The following assertion hold:
for all $\vare>0$, there exists a word $w$ such that $\cala(w) > 1-\vare$ 
iff $x >\frac{1}{2}$.
\end{lemma}
\begin{proof}
Given $\vare>0$, a word $w$ to be accepted with probability at least 
$1-\frac{\vare}{2}$ both $\cala_1$ and $\cala_2$ must accept it with
probability at least $1-\vare$.
Conversely, given a word $w$ if it is accepted by $\cala_1$ and $\cala_2$
with probability at least $1-\vare$, then $\cala$ accepts it with 
probability at least $1-\vare$.
Hence we show that for all $\vare>0$ there exist words $w$ such that 
both $\cala_1(w) \geq 1-\vare$ and $\cala_2(w)\geq 1-\vare$ iff $x > \frac{1}{2}$.

We first observe that a word $w$ to be accepted in $\cala_1$ can have a 
$\$$ only after having reached its absorbing accepting state (since for all 
other states the input letter $\$$ leads to the non-accepting 
absorbing state $q_4$). 
A word $w$ to be accepted in $\cala_2$ must have a $\$$, and the $\$$ must 
occur when the current state is either $q_5$ or $q_6$.
A word to be accepted in $\cala_1$ must contain a $b$ and the $b$ must occur 
when the current state is $q_1$.
Hence it suffices to consider words of the form $w^*=w w'$, where 
$w'$ is an arbitrary infinite word and $w$ is a finite word  of the form 
$a^{n_0}b a^{n_1} b \ldots a^{n_i} b \$$.
Consider a word $a^n b$, for $n \geq 0$: the probability to reach the 
state $q_3$ from $q_1$ is $x^n$, and the probability to reach the absorbing 
non-accepting state $q_7$ from $q_5$ is $(1-x)^n$.
Consider a word 
$w=a^{n_0}b a^{n_1} b \ldots a^{n_i} b \$$:
(a)~the probability of reaching $q_3$ from $q_1$ in $\cala_1$ is 
$1-\prod_{k=0}^i(1-x^{n_i})$;
(b)~the probability of reaching $q_7$ from $q_5$ in $\cala_2$ is  
\[
\begin{array}{rcl}
(1-x)^{n_1} + (1-(1-x)^{n_1})\cdot(1-x)^{n_2} + & \cdots & + \prod_{k=0}^{i-1}(1-(1-x)^{n_k})\cdot (1-x)^{n_i} \\[1ex]
& = & 1-\prod_{k=0}^i(1-(1-x)^{n_i}).
\end{array}
\]
If $x \leq \frac{1}{2}$, then $x \leq 1-x$, and hence 
$1-\prod_{k=0}^i(1-x^{n_i}) \leq 1-\prod_{k=0}^i(1-(1-x)^{n_i})$.
It follows that the acceptance probability of $w$ in $\cala_1$ is 
less than the rejection probability in $\cala_2$, hence 
$\cala(w) \leq \frac{1}{2}$.
It follows that if $x < \frac{1}{2}$, then for all words $w$ we 
have $\cala(w) \leq \frac{1}{2}$.

To complete the proof we show that if $x > \frac{1}{2}$, 
then for all $\vare>0$, there exist words $w$ such that 
both $\cala_1(w) \geq 1-\vare$ and $\cala_2(w)\geq 1-\vare$.
Our witness words $w^*$ will be of the following form: let 
$w=a^{n_0}b a^{n_1} b \ldots a^{n_i} b$ and we will have $w^*=w\$^\omega$.
Our witness construction closely follows the result of~\cite{GO-Tech}.
Given the word $w$, the probability to reach $q_3$ from $q_1$ is 
\[
L_1=1-\prod_{k=0}^i(1-x^{n_i});
\]
hence $\cala_1$ accepts $w^*$ with probability $L_1$.
Given the word $w$, since the last letter is a $b$, with probability~1 
the current state is either $q_7$ or $q_5$.
The probability to reach $q_7$ from $q_5$ for $w$ is given by 
\[
\begin{array}{rcl}
L_2 & = & (1-x)^{n_1} + (1-(1-x)^{n_1})\cdot(1-x)^{n_2} + \cdots + 
\prod_{k=0}^{i-1}(1-(1-x)^{n_k})\cdot (1-x)^{n_i} \\[1ex]
& \leq & \sum_{k=0}^i(1-x)^{n_i} \qquad \text{ since } 
\prod_{k=0}^{j}(1-(1-x)^{n_k}) \leq 1 \text{ for all } j \leq i.
\end{array}
\]
Thus the acceptance probability for $w$ in $\cala_2$ is at least $1-L_2$.
Hence given $\vare>0$, our goal is to construct a 
sequence $(n_0,n_1, n_2, \ldots,n_i)$ such that $L_1 \geq 1-\vare$ 
and $L_2 \leq \vare$.
For $\vare>0$, it suffices to construct an infinite sequence 
$(n_k)_{k \in \nats}$ such that 
\[
\prod_{k\geq 0}(1-x)^{n_k} =0; \quad \text{and} \quad
\sum_{k\geq 0}(1-x)^{n_k} \leq \vare.
\]
Then we can construct a finite sequence $(n_0,n_1,\ldots,n_i)$ of numbers 
such that we have both 
\[
\prod_{k =0}^i(1-x)^{n_k} \leq \vare; \quad \text{and} \quad
\sum_{k =0}^i(1-x)^{n_k} \leq 2\cdot\vare,
\]
as desired.
Since $\prod_{k\geq 0}(1-x)^{n_k} =0$ iff 
$\sum_{k=0}^\infty x^{n_k}=\infty$,
we construct sequence $(n_k)_{k \in \nats}$ such that 
$\sum_{k=0}^\infty x^{n_k}=\infty$,
and $\sum_{k\geq 0}(1-x)^{n_k} \leq \vare$.
Let $(n_k)_{k \in \nats}$ be a sequence such that $n_k=\ln_x(\frac{1}{k}) +J$, 
where $J$ is a suitable constant (to be chosen later and depends on 
$\vare$ and $x$).
Then we have
\[
\sum_{k \geq 0} x^{n_k} = x^J \cdot \sum_{k\geq 0}\frac{1}{k}=\infty.
\]
On the other hand, we have 
\[
1-x =x^{\ln_x(1-x)} =x^{\frac{\ln(1-x)}{\ln x}}
\]
Since $x>\frac{1}{2}$ we have $\beta=\frac{\ln(1-x)}{\ln x} > 1$, i.e.,
there exists $\beta>1$ such that $1-x=x^\beta$.
Hence 
\[
\sum_{k \geq 0} (1-x)^{n_k} =
\sum_{k \geq 0} x^{\beta\cdot n_k} 
=x^{\beta\cdot J} \cdot \sum_{k \geq 0} x^{\beta \cdot \ln_x(\frac{1}{k})}
=x^{\beta \cdot J} \cdot \sum_{k\geq 0} \frac{1}{k^\beta}.
\]
Since the above series converges, for all $\vare>0$, there exists a $J$ 
such that $\sum_{k \geq 0} (1-x)^{n_k} \leq \vare$.
This completes the proof.
\qed
\end{proof}

\smallskip\noindent{\bf Almost and limit problem do not coincide.} 
In contrast to probabilistic safety automata where
the almost and limit question coincide, the answers are different 
for probabilistic acceptance-absorbing reachability automata. 
In the automaton $\cala$ above if $x>\frac{1}{2}$, the answer to 
the limit question is yes. 
It follows from the proof above that if $\frac{1}{2} < x <1$, 
then though the answer to the limit question is yes, the 
answer to the almost question is no.
The almost and positive problem for reachability 
automata is decidable by reduction to POMDPs with reachability 
objective.
We show that the limit question is undecidable.
The proof uses the above lemma and the reduction technique of~\cite{GO-Tech}.

\smallskip\noindent{\bf Reduction for undecidability.}
Given an acceptance-absorbing reachability automaton $\calb$, we construct 
an automaton 
as shown in Fig~\ref{fig:limitundec}: we assume that $\sharp$ does not belong 
to the alphabet of $\calb$, and in the picture the dashed transitions 
from $\calb$ on $\sharp$ are from accepting states, and the solid transitions
are from non-accepting states.
We call the automaton as $\cala^*$.
We claim the following: the answer to the limit problem for the automaton 
$\cala^*$ is yes iff there exists a word $w$ that is accepted in $\calb$ with 
probability greater than $\frac{1}{2}$.
\begin{enumerate}

\item Suppose there is a word $w^*$ such that $\calb(w^*)=y>\frac{1}{2}$.
Let $\eta=y-\frac{1}{2}$. There is a finite prefix $w$ of $w^*$ such 
that the probability to reach an accepting state in $\calb$ given $w$ is 
$x=y-\frac{\eta}{2} > \frac{1}{2}$.
Given $\vare>0$ and $x$ as defined $y-\frac{\eta}{2}$, we have 
$x > \frac{1}{2}$, 
and hence by Lemma~\ref{lemm:limit_lemma} there exists a sequence  
$(n_0,n_1,\ldots,n_k)$ such that 
$(a^{n_0} b a^{n_1} b \ldots a^{n_k} b) (\$)^\omega$ is accepted in $\cala$ 
with probability at least $1-\vare$.
It follows that the automaton $\cala^*$ accepts 
$\big((a w\sharp)^{n_0} b (a w\sharp)^{n_1}\ldots b (a w\sharp)^{n_k} b\big) 
(\$)^\omega$ 
with probability at least $1-\vare$.
It follows that the answer to the limit problem is yes.

\item If for all words $w$ we have $\calb(w)\leq \half$, then we show that 
the answer to the limit problem is no. 
Let $x=\sup_{w \in \alphab^\omega} \calb(w) \leq \frac{1}{2}$.
Consider a word 
$\wh{w}=\big((a w_0\sharp)^{n_0} b (a w_1\sharp)^{n_1}\ldots b (a w_2\sharp)^{n_k} b \big)(\$)^\omega$:
given this word the probability to reach $q_3$ in $q_1$ given $\wh{w}$ is 
at most 
\[
1-\prod_{k=0}^i(1-x^{n_i});
\]
and the probability to reach $q_7$ from $q_5$ given $\wh{w}$ is at least 
\[
\begin{array}{rcl}
(1-x)^{n_1} + (1-(1-x)^{n_1})\cdot(1-x)^{n_2} + \cdots & + & \displaystyle \prod_{k=0}^{i-1}(1-(1-x)^{n_k})\cdot (1-x)^{n_i} \\[1ex]
& = & \displaystyle 1-\prod_{k=0}^i(1-(1-x)^{n_i}).
\end{array}
\]
The argument is similar to the proof of Lemma~\ref{lemm:limit_lemma}.
It follows from the argument of Lemma~\ref{lemm:limit_lemma} that the word 
is accepted with probability at most $\frac{1}{2}$ in $\cala^*$.
Thus it follows that for all words $w$ we have $\cala^*(w) 
\leq \frac{1}{2}$.
Since the quantitative existence problem is undecidable for acceptance 
absorbing reachability automata (Theorem~\ref{thrm2}), 
it follows that the limit problem is also undecidable.
\end{enumerate}

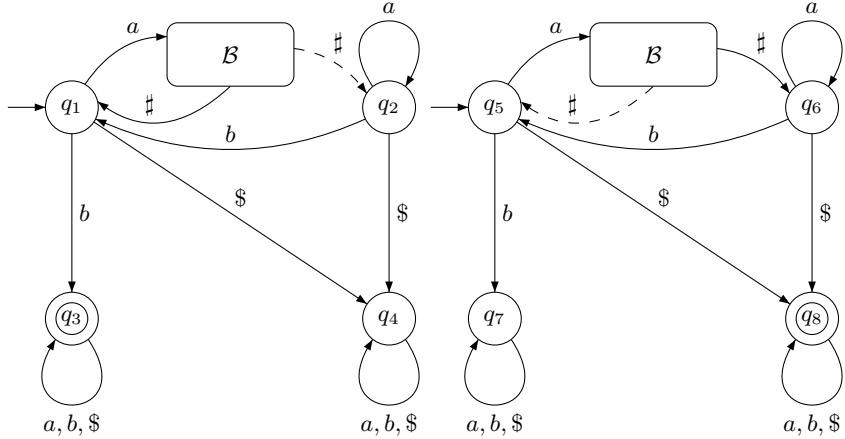
\begin{figure}[!tb]
  \begin{center}
    \unitlength=4pt
    \begin{picture}(25,50)(20,-30)
    \gasset{Nw=5,Nh=5,Nmr=2.5,curvedepth=4}
    \thinlines
    \node[Nmarks=i,iangle=180](A1)(0,0){$q_1$}
    \node(A2)(30,0){$q_2$}
    \drawloop[loopangle=90](A2){$a$}
    \node[Nmarks=r,rdist=1](A3)(0,-20){$q_3$}
    \node(A4)(30,-20){$q_4$}
    \node[Nw=12,Nh=6,Nmr=1](B1)(15,5){$\calb$} 
    \drawloop[loopangle=270](A3){$a,b,\$$}
    \drawloop[loopangle=270](A4){$a,b,\$$}
    \drawedge(A1,B1){$a$}
    \drawedge[syo=0,eyo=-3,dash={1.0}0](B1,A2){$\sharp$} 
    \drawedge[syo=-3,eyo=3,ELside=r](B1,A1){$\sharp$}
    \drawedge[ELside=r](A2,A1){$b$}
    \gasset{curvedepth=0}
    \drawedge(A1,A4){$\$$}
    \drawedge(A2,A4){$\$$}
    \drawedge(A1,A3){$b$}

    \gasset{Nw=5,Nh=5,Nmr=2.5,curvedepth=4}
    \thinlines
    \node[Nmarks=i,iangle=180](A5)(40,0){$q_5$}
    \node(A6)(70,0){$q_6$}
    \drawloop[loopangle=90](A6){$a$}
    \node(A7)(40,-20){$q_7$}
    \node[Nmarks=r,rdist=1](A8)(70,-20){$q_8$}
    \node[Nw=12,Nh=6,Nmr=1](B2)(55,5){$\calb$} 
    \drawloop[loopangle=270](A7){$a,b,\$$}
    \drawloop[loopangle=270](A8){$a,b,\$$}
    \drawedge(A5,B2){$a$}
    \drawedge[syo=0,eyo=-3](B2,A6){$\sharp$} 
    \drawedge[syo=-3,eyo=3,ELside=r,dash={1.0}0](B2,A5){$\sharp$}
    \drawedge[ELside=r](A6,A5){$b$}
    \gasset{curvedepth=0}
    \drawedge(A5,A8){$\$$}
    \drawedge(A6,A8){$\$$}
    \drawedge(A5,A7){$b$}

    \end{picture}
  \end{center}
  \caption{Limit undecidability construction.}\label{fig:limitundec}
\end{figure}

\begin{table}
\begin{center}
\begin{tabular}{|c|c|c|c|c|}
\hline
             & Positive & Almost   & Limit  & Quantitative  Equality \\ 
\hline
Safety &  
Dec. & Dec. & Dec. & Undec. \\ 
\hline
Reachability  &    
Dec. & Dec. & Undec. & Undec. \\ 
\hline
B\"uchi   &    
Undec. & Dec. & Undec. & Undec. \\ 
\hline
coB\"uchi  &   
Dec. & Undec. & Undec. & Undec. \\ 
\hline
Limit-average    &    
Dec. & Undec. & Undec. & Undec. \\
\hline
\hline
\end{tabular}
\end{center}
\caption{Decidability and undecidability results for probabilistic automata. 
}\label{tab2}
\vspace{-2em}
\end{table}

\smallskip\noindent{\bf Undecidability.} From the above proof it 
follows that the limit problem is undecidable for acceptance 
absorbing reachability automata, and hence also for B\"uchi, coB\"uchi, 
and limit-average condition.

\smallskip\noindent{\bf The case of Limit-average Automata.} The proofs for the decidability and undecidability results concerning limit-average conditions differ from the other conditions. We present how to use the results of \cite{tracol2009recurrence} and \cite{tracol-recurrence} to prove that the positive problem for limit-average automata is PSPACE-complete, and the almost problem for limit-average automata is undecidable. Given $\pi\in\Pi$, the authors of \cite{tracol2009recurrence} define:
\[\mathrm{Supp}(\pi)=\lbrace q\in Q\ s.t.\ \mathrm{LimitAvg}(\pi)>0\rbrace\]
In paper \cite{tracol2009recurrence} and \cite{tracol-recurrence} are considered the following decision problems:

\begin{itemize}
 \item \textbf{Positive Strong recurrence Problem: } Given a PA $\mathcal{A}$, decide whether there exists $w\in\Sigma^ \omega$ such that $\mathbb{P}^ w_\mathcal{A}[\lbrace \pi\in\Pi |\ \mathrm{Supp}(\pi)\cap F\not=\emptyset\rbrace]>0$.
 \item \textbf{Almost Sure Strong co-recurrence Problem: } Given a PA $\mathcal{A}$, decide whether there exists $w\in\Sigma^ \omega$ such that $\mathbb{P}^ w_\mathcal{A}[\lbrace \pi\in\Pi\ |\ \mathrm{Supp}(\pi)\cap \overline{F}=\emptyset\rbrace]=1$.

\end{itemize}

The first point of the following theorem was shown in \cite{tracol-recurrence}, and the second in \cite{tracol2009recurrence}, an extended version of \cite{tracol-recurrence}.
\begin{theorem}\label{t-res rec}
The Positive Strong recurrence Problem is PSPACE-complete, and the Almost Sure Strong co-recurrence Problem is undecidable.
\end{theorem}

By definition, for any PA $\mathcal{A}$ and any $w\in\Sigma^\omega$, for all $\pi\in\Pi$, we have $\mathrm{Supp}(\pi)\cap F\not=\emptyset$ iff $\mathrm{LimitAvg}(\pi)>0$. Also, for any PA $\mathcal{A}$ and any $w\in\Sigma^\omega$, for all $\pi\in\Pi$, we have that $\mathrm{Supp}(\pi)\cap \overline{F}=\emptyset$ iff $\mathrm{Supp}(\pi)\subseteq F$, iff $\mathrm{LimitAvg}(\pi)=1$. Finally, we get:
\begin{itemize}
 \item The following problem is PSPACE complete: given a PA $\mathcal{A}$, decide whether there exist $w\in\Sigma^ \omega$ such that $\mathbb{P}^ w_\mathcal{A}[\lbrace \pi\ |\ \mathrm{LimitAvg}(\pi)>0=\emptyset\rbrace]>0$.
 \item The following problem is undecidable: given a PA $\mathcal{A}$, decide whether there exist $w\in\Sigma^ \omega$ such that $\mathbb{P}^ w_\mathcal{A}[\lbrace \pi\ |\ \mathrm{LimitAvg}(\pi)=1\rbrace]=1$.
\end{itemize}


\begin{corollary}
The positive problem for limit-average automata is PSPACE-complete, and the almost problem for limit-average automata is undecidable
\end{corollary}
\begin{proof}
Concerning the first point, a standard result of probability theory says that on a probability space $(\Omega,\mathcal{F},\mu)$, given a function $f:\Omega\rightarrow\mathbb{R}^+$, we have:
\[\mathbb{E}^\mu[f]>0\ \mathrm{iff}\ \mu(\lbrace \pi\in\Omega\ s.t.\ f(x)>0\rbrace)>0\]
Given $w\in\Sigma^\omega$, taking $\Omega$ the set of paths $\Pi$, $\mathcal{F}$ as the $\sigma$-algebra generated by cones on $\Omega$, $\mu$ as $\mathbb{P}^w_\mathcal{A}$ and $f$ as the $\mathrm{LimitAvg}$ function, we get that:
\[\mathbb{E}^w_\mathcal{A}[\mathrm{LimitAgv}]>0\ \mathrm{iff}\ \mathbb{P}^w_\mathcal{A}[\lbrace \pi\in\Pi\ s.t.\ \mathrm{LimitAvg}(\pi)>0\rbrace]>0\]
Hence the Positive Limit-average Problem is equivalent to the Positive Strong recurrence Problem, which is PSPACE-complete.

Let us now consider the second point. Let $\overline{\mathrm{LimitAvg}}:\Pi\rightarrow\lbrace0,1\rbrace$ be such that for all $\pi\in\Pi$ we have $\overline{\mathrm{LimitAvg}}(\pi)=1-\mathrm{LimitAvg}(\pi)$. By the same argument as before, given $w\in\Sigma^ \omega$, we have:
\begin{equation}\label{e-first}
 \mathbb{E}^w_\mathcal{A}[\overline{\mathrm{LimitAvg}}]>0\ \mathrm{iff}\ \mathbb{P}^w_\mathcal{A}[\lbrace \pi\in\Pi\ s.t.\ \overline{\mathrm{LimitAvg}}(\pi)>0\rbrace]>0
\end{equation}

Now, $\mathbb{E}^w_\mathcal{A}[\overline{\mathrm{LimitAvg}}]=1-\mathbb{E}^w_\mathcal{A}[\mathrm{LimitAvg}]$. As a consequence, we have that $\mathbb{E}^w_\mathcal{A}[\mathrm{LimitAvg}]=1$ iff $\mathbb{E}^w_\mathcal{A}[\overline{\mathrm{LimitAvg}}]=0$, hence by equation \ref{e-first} 
\[\mathbb{E}^w_\mathcal{A}[\mathrm{LimitAvg}]=1\ \mathrm{iff}\  \mathbb{P}^w_\mathcal{A}[\lbrace \pi\in\Pi\ s.t.\ \overline{\mathrm{LimitAvg}}(\pi)>0\rbrace]=0\]
Which is also equivalent to $\mathbb{P}^w_\mathcal{A}[\lbrace \pi\in\Pi\ s.t.\ \overline{\mathrm{LimitAvg}}(\pi)=0\rbrace]=1$, and to $\mathbb{P}^w_\mathcal{A}[\lbrace \pi\in\Pi\ s.t.\ \mathrm{LimitAvg}(\pi)=1\rbrace]=1$. The remark after theorem \ref{t-res rec} complete then the proof.

\end{proof}

\smallskip\noindent{\bf Summary of results.} We have seen that the limit problem is undecidable for acceptance 
absorbing reachability automata, and hence also for B\"uchi, coB\"uchi, 
and limit-average condition.

This gives us the Theorem~\ref{thrm3}.
For other qualitative questions the results are as follows: 
(A) We argued that for safety condition that almost and limit problem 
coincide, and the decision problems are PSPACE-complete~\cite{chadha2009expressiveness}.
(B) For reachability condition the almost and positive problem 
can be answered through POMDPs (\cite{BBG08}), and the almost problem was 
shown to be PSPACE-complete~\cite{chadha2009expressiveness}.
(C) For B\"uchi condition it was shown in~\cite{BBG08} that the 
positive problem is undecidable and the almost problem is decidable in EXPTIME.
The almost problem was later shown to be PSPACE-complete~\cite{chadha2009power}.
(D) For coB\"uchi condition it was shown in~\cite{BBG08} that the 
almost problem is undecidable and the positive problem is decidable through solution 
of POMDPs with coB\"uchi condition in EXPTIME~\cite{BGG09}.
The positive problem was later shown to be PSPACE-complete in~\cite{chadha2009power}.
(E) For limit-average condition the results follow from the last paragraph.
For the exact arithmetic hierarchy characterization of the undecidable problems of \cite{BBG08} 
see \cite{chadha2009expressiveness,chadha2009power}. The results are summarized in Table~\ref{tab2} 
(the results for quantitative existence and quantitative value problem is the same as 
for quantitative equality problem).

\begin{theorem}[Qualitative problem]\label{thrm3}
The following assertions hold:
(a)~the positive problem is decidable for probabilistic safety, reachability, 
coB\"uchi and limit-average automata, and is undecidable for probabilistic B\"uchi automata;
(b)~the almost problem is decidable for probabilistic safety, reachability, 
and B\"uchi automata, and is undecidable for probabilistic coB\"uchi and limit-average automata;
and
(c)~the limit problem is decidable for probabilistic safety automata and
is undecidable for probabilistic reachability, B\"uchi, coB\"uchi and 
limit-average automata.
\end{theorem}

\end{document}